%% file: kl-ocolr-05t.tex
\documentclass[a4paper,11pt]{article}
\usepackage[dvips]{color}
\usepackage{paper,latexsym,multicol}
\usepackage{trbonn}
\usepackage{changebar}

\def\qed{}

\newcommand{\agent}{robot}
\let\epsilon=\varepsilon
\newcommand{\ie}{i.\,e.}

\newcommand{\DP}{\ell^+}
\newcommand{\DM}{\ell^-}
\newcommand{\etal}{et~al.\ }

\def\proofskname{Proof Sketch}
        {\proofendbox\par\smallskip}
\newtheorem{proposition}{Proposition}
\newenvironment{prop}[1]{\begin{proposition}\label{#1-prop}}%
{\end{proposition}}

\title{Optimal competitive online ray search with an error-prone robot}

\author{Tom Kamphans \and Elmar Langetepe}

\date{February 2005}
\reportnumber{003}

\begin{document}

\maketitle

\begin{abstract}
We consider the problem of finding a door along a wall with a blind
robot that neither knows the distance to the door nor 
the direction towards of the door. 
This problem can be solved with the well-known doubling
strategy yielding an optimal competitive factor of 9 with the 
assumption that the robot does not make any errors during
its movements. We study the case that the robot's movement
is erroneous. In this case the doubling strategy is no longer optimal. 
We present optimal competitive strategies that take the error assumption 
into account. The analysis technique can be applied to different error models. 

\smallskip\noindent
{\bf Keywords:} Online algorithm, Online motion planning, competitive ana\-lysis, ray search, errors.
\end{abstract}

\input{intro}
\input{result}

\nocite{py-spwm-91}
\nocite{k-maole-05}

\bibliographystyle{abbrv}
\bibliography{../../../abt1/biblio/local,%
              ../../../abt1/biblio/update,%
              ../../../abt1/biblio/geom}

\clearpage
\input{appendix}

\end{document}

%% file: intro.tex
\section{Introduction\label{intro-sect}}
Motion planning in unknown environments 
is theoretically well-understood and also practically solved 
in many settings. During the last decade many different objectives
where discussed under several robot models. For a general overview of
theoretical online motion planning problems and its analysis 
see the surveys \cite{b-olsn-98,m-gspno-00,ikkl-ccnt-02,rksi-rnut-93,rksi-rnut-93}

Theoretical correctness results and performance guarantees 
often suffer from idealistic assumptions so that in the worst case a 
correct implementation is impossible. On the 
other hand practioners analyze correctness results and 
performance guarantees mainly statistically or empirically. 
Therefore it is useful to investigate, 
how theoretic online algorithms with idealistic assumptions behave, 
if those assumptions cannot be fulfilled. More precisely, 
can we incorporate assumptions of errors in sensors and motion 
directly into the theoretical sound analysis? We already
successfully considered the behaviour of the well-known pledge algorithm,
see e.\,g.\ Abelson and diSessa \cite{ad-tg-80},
and Hemmerling \cite{h-lp-89},
in the presence of errors \cite{kl-pares-03}.
It has been proven that 
a compass with an error of $\pi/2$ is sufficient to leave an unknown 
maze with a strategy that makes use of counting turning angles. 

The task of finding a point on a line by a blind agent 
without knowing the location of the 
goal was considered by Gal~\cite{g-sg-80,g-csg-89,ag-tsgr-03}
and independently reconsidered by Baeza-Yates \etal \cite{bcr-sp-93}.  
Both approaches lead to the so called doubling strategy, which is a basic paradigma 
for searching algorithms, e.\,g. searching for a point on $m$ rays, see
\cite{g-sg-80}, or approximating the optimal search path, see \cite{fkklt-colao-04}.  

Searching on the line was generalized to searching on $m$ concurrent rays 
starting from a single source, see \cite{bcr-sp-93,g-sg-80}. 
Many other variants were discussed since then, for example
$m$-ray searching with restricted distance (Hipke \etal \cite{hikl-hfplf-99}, 
Langetepe \cite{l-dasas-00}, Schuierer \cite{ls-ussmr-01}), $m$-ray 
searching with additional turn costs (Demaine \etal \cite{dfg-olstc-}),
parallel $m$-ray searching (Hammar \etal \cite{hns-psmr-01}) or
randomized searching (Kao \etal \cite{krt-sueor-96}).
Furthermore, some of the problems were again rediscovered 
by Jaillet \etal \cite{js-ols-01}.

In this paper we investigate how an error in the movement 
influences the correctness and the corresponding competitive 
factor of a strategy. 
The error range, denoted by a parameter $\delta$,  
may be known or unknown to the strategy. 

The paper is organized as follows. In \refsect{doubling} we recapitulate some 
details on $m$-ray searching and its analysis. The error model is 
introduced in \refsect{model}.
In \refsect{unknown} we discuss the case where the strategy is not aware 
of errors, therefore we
analyze the standard doubling strategy 
showing correctness and performance results.
The main result is presented in \refsect{known}. We can prove that 
the optimal  competitive strategy that searches for a goal on a line 
achieves a factor of $1+8\left(\frac{1+\delta}{1-\delta}\right)^2$
if the error range $\delta$ is known. Fortunately, our analysis technique works for 
different error models and is generic in this sense.
Finally, we consider the $m$-ray searching in \refsect{mray}. 
For a summary of the results and factors see \refsect{concl}. 
A preliminary version of this report appeared in \cite{kl-ocolr-05}.

%% file: result.tex
\sect{The standard problem}{doubling}

The task is to find a door in a wall, respectively a point, $t$, on a line. 
The \agent\ does not know whether $t$ is located left hand or
right hand to its start position, $s$, nor does it know the distance from
$s$ to $t$. Baeza-Yates \etal \cite{bcr-sp-93} describe a strategy for solving 
this problem by using a function $f$. $f(i)$ denotes the distance the
\agent\ walks in the $i$-th iteration. If $i$ is even, the \agent\ moves
$f(i)$ steps from the start to the right and $f(i)$ steps back; if $i$ is
odd, the \agent\ moves to the left. It is assumed that the movement
is correct, so after moving $f(i)$ steps from the start point to the right 
and moving $f(i)$ steps to the left, the \agent\ has reached the start
point.
Note, that this does not hold, if there are errors in the movement, 
see for example \reffig{leftdrift}. 

The competitive analysis compares the cost of a strategy to
the cost of an optimal strategy that knows the whole environment.
In our case these cost is given by the 
distance, $d$, to the goal.
With the assumption $d\geq 1$
a search strategy that generates a path of 
length $|\pi_{\mathrm{onl}}|$ 
is called {\em $C$-competitive} if for all possible scenarios
${|\pi_{\mathrm{onl}}| \over d} \leq C$ holds%
\footnote{For $d\geq 1$ the constant that usually appears in the 
definition of the competitive factor can be omitted
see for example \cite{l-dasas-00}.}.
It was shown by Baeza-Yates \etal \cite{bcr-sp-93} that the strategy 
$f(i)=2^i$ yields a competitive factor of 9 and that no other strategy will 
be able to achieve a smaller factor. 

The problem was extended to $m$ concurrent rays. 
It was shown by Gal \cite{g-csg-89} that w.l.o.g.\ a strategy visits the rays 
in a cyclic order and with increasing distances $f(i)<f(i+1)$. The 
optimal competitive factor is given by 
$1+ 2 \frac{m^m} {(m-1)^{m-1} }$ 
and an optimal strategy is defined by $f(i)=\left(\frac{m}{m-1}\right)^i$, see \cite{bcr-sp-93,g-sg-80}

\ssect{Modelling the error}{model}
The \agent\ moves straight line segments of a certain length from the
start point alternately to the left and to the right. Every 
movement can be errorneous, which causes the robot to
move more or less far than expected. However, we require that the \agent s
error per unit is within a certain error bound, $\delta$. 
More precisely, let $f$ denote the length of a movement required by
the strategy---the nominal value--- 
and let $\ell$ denote the actually covered distance,
then we require that
$\ell \in [\, (1-\delta)f, (1+\delta)f ]$ holds for
$\delta\in[\,0,1[$, \ie\ the robot moves at least $(1-\delta)f$
and at most $(1+\delta)f$.
This is a reasonable error model, since
the actually covered distance is in a symmetrical range around 
the nominal value. Another commonly used method is to require
$\ell \in [\,\frac1{1+\delta'}f, (1+\delta')f]$ 
for $\delta'>0$.
This leads to an unsymmetrical range around the desired value, but
does not restrict the upper bound for the error range.
Since both error models may be of practical interest, we give results for
both models, although we give full proofs only for the first model.
We call the first model {\em percentual error}, the second
model {\em standard multiplicative error}. 

\sect{Disregarding the error}{unknown}

In this section, we assume that the \agent\ is not aware of making any
errors. 
Thus, the optimal doubling strategy presented above 
seems to be the best choice for the \agent. In the 
following we will analyze the success and worst case 
efficiency of this strategy with respect to the unknown 
$\delta$. 

\pstexfig{The $i$-th iteration consists of two separate movements,
$\DP_i$ and $\DM_i$. Both may be of different length, causing a 
drift. The vertical path segments are to highlight the single 
iterations, the robot moves on horizontal segments only.}
{drifterror}

Since the errors in the movements away from the door and back towards
the door may be different, the robot may not return to the start point, $s$,
between two interations, see \reffig{drifterror}.
Even worse, the start point of every iteration may continuously 
drift away from the original start point.
Let $\DP_i$ be the length of the movement to the right 
in the $i$-th step and $\DM_i$ be the covered distance to the left.  
Now, the deviation from the start point after the 
$k$-th iteration step, the {\em drift} $\Delta_k$, is 
$$\Delta_k = \sum_{i=1}^k(\DM_i - \DP_i)\,.$$

If the drift is greater than zero, the start point $s_{k+1}$
of the iteration $k+1$ is located left to the original start point,
if it is smaller than zero, $s_{k+1}$ is right hand to $s$. 
Note, that $\DP_i$ equals $\DM_i$ in the error-free case.
We will show that the worst case is achieved, if the robot's drift to the
left is maximal.
The length of the path $\pi_k$ after $k$ iterations is
$$|\pi_k|=\sum_{i=1}^k(\DM_i + \DP_i)\,.$$

\begin{theo}{wc}
In the percentual error model
$[(1-\delta)f, (1+\delta)f]$ with $\delta\in[0,1[$ the \agent\ will find the door with the doubling strategy $f(i)=2^i$, 
if the error $\delta$ is not greater than ${1\over 3}$. The generated path
is never longer than 
$$8 \frac{1+\delta}{1-3\delta}+1$$
times the shortest path to the door%
\footnote{More precisely, the factor is 
$1+8 \frac{1+\delta}{1-3\delta+\varepsilon}$ for an arbitrary small
$\varepsilon$, which is crucial for the case $\delta=\frac13$, but
neglectable in all other cases. For convenience we
omit the $\varepsilon$ in this and the following theorems.}.
\end{theo}

\pstexfig{In the worst case, the start point of every iteration drifts
away from the door. }{leftdrift}

\begin{proof}
We assume that finally the goal is found on the right side. 
The other case is handled analogously. 
For the competitive setting it is the worst, 
if the door is hit in the iteration step ${2j+2}$ to the right side, 
but located just a little bit further away than the rightmost point that 
was reached in the preceeding iteration step~${2j}$. 
In other words, another full iteration to the left has to be done and 
the shortest distance to the goal is minimal in the current situation. 
Additionally we have to consider the case where the goal is exactly one step away from 
the start. We discuss this case at the end. 

\noindent
We want the door to be located closely behind the rightmost point visited 
in the iteration step $2j$. 
Considering the drift $\Delta_{2j-1}$, the distance from the 
start point $s$ to the door is
$$d = \DP_{2j} - \sum_{i=1}^{2j-1}(\DM_i - \DP_i) + \varepsilon\,.$$

\noindent
The total path length is the sum of the covered distances up to
the start point of the last iteration, the distance from this point
to the original start point $s$ (\ie\ the overall drift), and the distance to
the door:
$$|\pi_{\mathrm{onl}}| = \sum_{i=1}^{2j+1}(\DM_i + \DP_i) + \sum_{i=1}^{2j+1}(\DM_i - \DP_i) + d.$$
With this we get the worst case ratio
\begin{equ}{wcequ}
\frac{|\pi_{\mathrm{onl}}|}{d} 
= 1 + \frac{\sum_{i=1}^{2j+1}(2\DM_i)}{\DP_{2j} - \sum_{i=1}^{2j-1}(\DM_i - \DP_i) + \varepsilon} \, .
\end{equ}

We can see that this ratio achieves its maximum if we maximize 
every $\DM_i$, \ie\ if we set $\DM_i$ to $(1+\delta)\,2^i$ in this
error model. 
Now we only have to fix $\DP_i$ in order to maximize the ratio. 
Obviously, the denominator gets its smallest value if every $\DP_i$ 
is as small as possible, therefore we set $\DP_i$ to $(1-\delta)\,2^i$.
So the worst case is achieved
if every step to the right is too short and
every step to left is too long, yielding a maximal drift to the left
and away from the door, see \reffig{leftdrift}%
\footnote{The vertical path segments are to highlight the single 
iteration steps, the robot moves only on horizontal segments.}.
Altogether, we get
\begin{eqnarray}
\frac{|\pi_{\mathrm{onl}}|}{d} 
&=& 1 \!+\! \frac{\sum_{i=1}^{2j+1}(2\DM_i)}{\DP_{2j} - \sum_{i=1}^{2j-1}
(\DM_i \!\!-\! \DP_i) + \varepsilon}
 = 1 \!+\!  {2\,(1+\delta)\,\sum_{i=1}^{2j+1}2^i \over (1-\delta)\,2^{2j} - 2\delta \sum_{i=1}^{2j-1}2^i + \varepsilon}\nonumber\\
&=& 1 + 2\,\frac{(1+\delta)(2^{2j+2}-2) } {(1-3\delta)\,2^{2j}+4\delta + \varepsilon}\\
&<& 1+ 8\,\frac{1+\delta}{1-3\delta}\, .\label{ratio}
\end{eqnarray}
For the case that the goal is exactly one step away from the start we achieve 
the worst case factor $1+4\frac{(1+\delta)}{1}$ which is smaller than the
above worst case. 

Obviously, we have to require that $\delta\leq\frac13$ holds. Otherwise, 
in the worst case the 
distance $(1-3\delta)\,2^{2j}+4\delta$ from the start point does not
exceed the point $4\delta$ 
and we will not hit any goal farther away.
\end{proof} 

\begin{prop}{first}
In the standard multiplicative error model 
$[\frac1{(1+\delta)}f, (1+\delta)f]$ for $\delta>0$, 
the doubling strategy always finds
the goal with
a competitive factor of 
$1+8\,\frac{(1+\delta)^2}{2-(1+\delta)^2}$ 
if $\delta\leq\sqrt{2}-1$ holds.
\end{prop}
\begin{proof} We exactly follow the  proof of \reftheo{wc}. The worst case ratio is given by
\refequ{wcequ} and now we maximize this value by  $\DM_i = 2^i(1+\delta)$ and 
 $\DP_i=\frac{2^i}{(1+\delta)}$. Using these settings in (\ref{ratio})
yields
\begin{eqnarray*}\displaystyle
{|\pi_{\rm onl}| \over d}
&\leq & 1+ {2\,\sum_{i=1}^{2j+1}(1+\delta)\,2^i \over
{1\over 1+\delta}\,2^{2j}-\sum_{i=1}^{2j-1}
\left(1+\delta-{1\over 1+\delta} \right)\,2^i + \epsilon}\\
&=& 1 + {2\,(1+\delta)^2\,(4-{2\over 2^{2j}}) \over
1 - ((1+\delta)^2- 1)
\left(1-{2\over 2^{2j}}\right)+{\epsilon\,(1+\delta)\over 2^{2j}} }\\
&=& 1 + {2\,(1+\delta)^2\,(4-{2\over 2^{2j}}) \over
2 - (1+\delta)^2 +
\left({2\over 2^{2j}}((1+\delta)^2-1)\right)+
{\epsilon\,(1+\delta)\over 2^{2j}} }\\
&<& 1+8\cdot {(1+\delta)^2 \over 2-(1+\delta)^2}\,,
\end{eqnarray*}
and 
we achieve a worst case  ratio  of $1+8\,\frac{(1+\delta)^2}{2-(1+\delta)^2}$.
For the denominator $2-(1+\delta)^2\leq 0$ 
holds iff  $\delta > \sqrt{2}-1$.\qed
\end{proof}

\sect{The optimal strategy for known error range}{known}

One might wonder whether there is a strategy which takes the error $\delta$ 
into account and yields a competitive factor smaller than the 
worst case factor of the doubling strategy. 
Intuitively this seems to be impossible, because the doubling strategy 
is optimal in the error-free case. 
We are able to show that 
there is a strategy that achieves a factor of $1+ 8
\left(\frac{1+\delta}{1-\delta}\right)^2$. This is 
smaller than $1+ 8 \frac{1+\delta}{1-3\delta}$ for all $\delta<1$. 

\medskip
\begin{theo}{first}
In the presence of an error up to $\delta$ 
in the percentual error model 
$[(1-\delta)f, (1+\delta)f]$ with $\delta\in[0,1[$,
there is a strategy that meets every goal and achieves 
a competitive factor of
$1+ 8 \left(\frac{1+\delta}{1-\delta}\right)^2$.
\end{theo}

\begin{proof} We are able to design a strategy 
as in the error-free case. We can assume that a strategy $F$ is given by a 
sequence of non-negative values%
\footnote{In the following we use $f_i$ instead of $f(i)$ to omit
parenthesis in the equations.},
$f_1,f_2,f_3,\ldots$, denoting the nominal values 
required by the strategy,
\ie\ in the $i$-th step the strategy wants the robot to move a distance
of $f_i$ to a specified direction--- to the right, if $i$ is even
and to left if $i$ is odd---and to return to the start point
with a movement of $f_i$ to the opposite direction.
Remark, that every reasonable strategy can be described in this way.

As above, let $\DP_i$ and $\DM_i$ denote the length of a 
movement to the right and to the left in the $i$-th step, respectively. 
In the proof of \reftheo{wc} we showed
that every online strategy will achieve a worst case ratio of
\[
\frac{|\pi_{\mathrm{onl}}|}{d} 
= 1 + \frac{\sum_{i=1}^{2j+1}(2\DM_i)}{\DP_{2j} - \sum_{i=1}^{2j-1}(\DM_i -
  \DP_i) + \varepsilon}\;,
\]
which achieves its maximum if every step towards the door is as short
as possible and every step in the opposite direction is as long as
possible, \ie\ $\DM_i=(1+\delta)\,f_i$ and $\DP_i=(1-\delta)\,f_i$.
This yields
\begin{equ}{simplify}
\frac{|\pi_{\mathrm{onl}}|}{d} 
= 1 + 2(1+\delta)\frac{\sum_{i=1}^{2j+1}f_i}{(1-\delta)f_{2j} - 2\delta\sum_{i=1}^{2j-1}f_i + \varepsilon}
\end{equ}

\noindent
For a fixed $\delta$, $1$ and  $2(1+\delta)$ are constant, and 
it is sufficient to find a strategy that minimizes 
\begin{equ}{func}
G_{(n,\delta)}(F):=\frac{\sum_{i=1}^{n+1}f_i}{(1-\delta)f_n - 2\delta\sum_{i=1}^{n-1}f_i}
\;\mbox{~if~}\,n\geq 1
\end{equ}%
and $G_{(0,\delta)}(F):=\frac{f_1}{1}$
where $F$ denotes the strategy $f_1,f_2,f_3,\ldots$  and  $G_{(0,\delta)}(F)$ 
refers
to the worst case after the first iteration step. Note, that we assumed that the goal is at least one step 
away from the start%
\footnote{Alternatively, we can assume that the cost in the start situation is 
subsumed by an additive constant in the definition of the competitive factor.}.

By simple analysis we found out that a strategy $f_i=\alpha^i$ 
asymptotically minimizes $G_{(n,\delta)}(F)$ if $\alpha=2\,\frac{1+\delta}{1-\delta}$ 
holds. In this case we have 
\begin{eqnarray}
 G_{(n,\delta)}(F) &=& \frac{\sum_{i=1}^{n+1}f_i}{(1-\delta)f_n - 2\delta\sum_{i=1}^{n-1}f_i}\label{funcfac-equ}\\
                 &=& 4 \frac{(1+\delta)\left(2\,\frac{1+\delta}{1-\delta}\right)^{n+1}-\frac{1-\delta}{2}}{(1-\delta)^2 \left(2\,\frac{1+\delta}{1-\delta}\right)^{n+1}+4\delta(1-\delta)} <  4 \frac{1+\delta}{(1-\delta)^2}\nonumber
\end{eqnarray}
which altogether gives
$1+8\left(\frac{1+\delta}{1-\delta}\right)^2$
for the competitive ratio. Note, that $f_i=\left(2\,\frac{1+\delta}{1-\delta}\right)^i$ 
is a reasonable strategy, although handi\-capped by the error it monotonically
increases the distance to the start which is never smaller than the 
denominator in (\ref{simplify-equ})
  $$ (1-\delta)f_n - 2\delta\sum_{i=1}^{n-1}f_i = \frac1{(1+3\delta)}\left((1-\delta)(1+\delta)\left(2\,\frac{1+\delta}{1-\delta}\right)^{n}+4\delta(1+\delta)\right)\, .$$
Thus, every goal point will be reached.
\end{proof}

We want to show that the given factor is optimal. 
As we have seen in the proof of \reftheo{first} there is a constant upper bound   
of $G_{(n,\delta)}(F)$ that depends on $\delta$.  Now $C(F,\delta):= \sup_n G_{(n,\delta)}(F)$ 
defines the competitive value of a strategy $F$, 
and for every fixed $\delta$ there will 
be a strategy $F^*$ that yields $C^*_\delta:=\inf_F C(F,\delta)$. In the following 
we want to show that there is always a strategy $F^*$ that achieves 
$C^*_\delta$ exactly in every step, that is $G_{(n,\delta)}(F^*)=C^*_\delta$ for $n\geq 1$.
For example, for $\delta=0$ the strategy $f_i=(i+1)2^i$ 
exactly yields the factor 9 in 
every step. Note, that up to now only the existence of $C^*_\delta$ is known. 
The idea of using equality was mentioned in \cite{kpy-sfg-96} and used 
for a finite strategy in \cite{ls-ussmr-01}. We give a formal proof for 
infinite strategies with errors.

\begin{lem}{exact} 
In the presence of an error up to $\delta$ in the percentual error model
$[(1-\delta)f, (1+\delta)f]$ with $\delta\in[0,1[$,
there is always an optimal strategy $F^*$ that achieves the optimal (sub)factor $C^*_\delta$ 
exactly for all $G_{(n,\delta)}(F^*)$, that is $G_{(n,\delta)}(F^*)=C^*_\delta$ for $n\geq 1$. 
\end{lem} 
\begin{proof} 
Let $F$ be a $C^*_\delta$ competitive strategy. We can assume that $F$ is 
strictly positive. We will show that for all $n\geq 1$ there 
is always a strategy $F'$ that fulfills $G_{(k,\delta)}(F')=C^*_\delta$ for all $1\leq k\leq n$, by successively adjusting $F$ adequately. 

First, note that 
$G_{(n,\delta)}(F)$ is decreasing in $f_n$. The value 
$(1-\delta)f_n$ $\mbox{}- 2\delta\sum_{i=1}^{n-1}f_i$ describes the 
distance to the start. We can assume that ${(1-\delta)f_n - 2\delta\sum_{i=1}^{n-1}f_i}>0$ holds, 
otherwise the strategy does not increase the distance to the start in this step and there is a 
better strategy for the worst-case. Now $G_{(n,\delta)}(F)$ equals a function $\frac{f_n+A}{C\cdot f_n-B}$ 
in $f_n$ with $A,B,C>0$ and ${Cf_n-B}>0$. The function   $\frac{f_n+A}{C\cdot f_n-B}$ is positive 
and strictly decreasing for ${Cf_n-B}>0$. The other way round, the function goes to infinity if 
$f_n$ decreases towards $f_n=\frac{B}{C}$. Altogether, if we decrease $f_n$ adequately we can 
increase $G_{(n,\delta)}(F)$ continously to whatever we want.  

Additionally, $G_{(k,\delta)}(F)$ is increasing in $f_n$ for all $k\neq n$. For $k< {n-1}$ 
$G_{(k,\delta)}(F)$ is not affected by $f_n$. For $k={n-1}$ the distance $f_n$ appears only 
in the nominator of  $G_{(k,\delta)}(F)$ which increases if $f_n$ grows.  
 For $k>n$ if $f_n$ grows, the denominator shrinks (it has the coefficient  $-2\delta$) and the nominator 
increases. The other way round, if $f_n$ decreases, $G_{(k,\delta)}(F)$ is 
decreasing in $f_n$ for all $k\neq n$. Altogether, if we decrease $f_n$ we will decrease 
all $G_{(k,\delta)}(F)$ for $k\neq n$. 

As indicated above let $F$ be a strictly positive $C^*_\delta$ competitive strategy. 
By induction over $n\geq 1$, we will show that we can decrease $F$ to a strictly positive 
strategy $F'$ which fulfills $G_{(k,\delta)}(F')=C^*_\delta$ for all $1\leq k\leq n$. 
Additionally, $F'$ equals $F$ for all $f_l$ and $l\geq n+1$. 
 
For $n=1$ let us assume that $G_{(1,\delta)}(F)<=C^*_\delta$ holds,
otherwise we are done. 
By the argumentation above we will decrease $f_1$ by a small value 
to $f_1'= f_1-\epsilon$ such that 
$G_{(1,\delta)}(F)=C^*_\delta$ holds. All other $G_{(k,\delta)}(F)$ with $k\neq 1$ will decrease, and
the new strategy is still $C^*_\delta$ competitive. Since $f_2$ is strictly positive, 
$f_1'$ has to be strictly positive. 

For the induction step, we assume that we decrease $F$ to a strictly positive strategy $F'$ such that 
$G_{(k,\delta)}(F')=C^*_\delta$ for all $1\leq k\leq n$. $F'$ equals $F$ for all $f_l$ and $l\geq n+1$.
Now let $G_{(n+1,\delta)}(F')<C^*_\delta$,
otherwise we are done. By the argumentation above we can 
decrease $f_{n+1}$ by $\varepsilon$ to $f_{n+1}'= 
f_{n+1}-\epsilon$ such that  $G_{(n+1,\delta)}(F')=C^*_\delta$ holds. 
With the considerations above we know that  $G_{(k,\delta)}(F')$ decreases for $k\neq n$. 
Since $f_{n+2}$ is strictly positive, $f'_{n+1}$ has to be strictly positive. 

Unfortunately, at least $G_{(n,\delta)}(F')<C^*_\delta$ holds and  we have to apply the induction hypothesis again. 
We decrease  $F'$ to $F''$ such that  $G_{(k,\delta)}(F'')=C^*_\delta$ holds for all $1\leq k\leq n$. 
Again $F''$ is strictly positive. Now in turn  we will have $G_{(n+1,\delta)}(F'')<C^*_\delta$ and 
we start the procedure again by decreasing $f'_{n+1}$ adequately. 
For convenience we denote $F''$ by $F'$ again. 

Altogether, for every $f'_l$ with $1\leq l\leq n+1$ in the above process we will have a strictly 
decreasing sequence which will not decrease towards zero. Therefore every  $f'_l$ runs toward a 
unique positive limit. For the limit values we will have $G_{(k,\delta)}(F')=C^*_\delta$ 
for $1\leq k\leq n+1$ which finishes the induction and the proof. \qed
\end{proof}

Fortunately, we can build a recurrence for the optimal strategy $F^*$ with $G_{(n,\delta)}(F^*)=C^*_\delta$
for $n\geq 3$. From $G_{(n-1,\delta)}(F^*)=G_{(n-2,\delta)}(F^*)=C^*_\delta$ 
we conclude
$\sum_{i=1}^{n}f^*_i=C^*_\delta\left((1-\delta)f^*_{n-1} - 2\delta\sum_{i=1}^{n-2}f^*_i\right)$ 
and  
$\sum_{i=1}^{n-1}f^*_i=C^*_\delta\left((1-\delta)f^*_{n-2} - 2\delta\sum_{i=1}^{n-3}f^*_i\right)$.
Subtracting both yields for all $n\geq 3$
\begin{equ}{rec}
f^*_{n}=C^*_\delta(1-\delta)f^*_{n-1}-C^*_\delta(1+\delta)f^*_{n-2}   
\end{equ}%
\begin{theo}{opt}
In the presence of an error up to $\delta$ in the percentual error model
$[(1-\delta)f, (1+\delta)f]$ with $\delta\in[0,1[$,
there is no competitive strategy for searching a point on a line 
that yields a factor smaller than  
$1+8\left(\frac{1+\delta}{1-\delta}\right)^2$. 
\end{theo}
\begin{proof} 
We solve the recurrence \refequ{rec} using methods described in 
Graham et.\,al.~\cite{gkp-cm-94}. The characteristic polynom of 
recurrence \refequ{rec} is given by 
\begin{equ}{char}
X^2-C^*_\delta(1-\delta)X+C^*_\delta(1+\delta)\,,
\end{equ} 
which has the roots 
\begin{equ}{sol1}
\lambda, \overline{\lambda}
  =  \frac12\left((1-\delta) C^*_\delta 
          \pm\sqrt{C^*_\delta\left(C^*_\delta(1-\delta)^2-4(1+\delta)\right)}\right)
\end{equ}%
where $\overline{\lambda}$ denotes the conjugate of $\lambda$. 
Now, the recurrence is given in the closed form 
$f^*_{n}=a \lambda^n+ \overline{a}\overline{\lambda}^n = 2\mbox{Re}(a\lambda^n)$ 
where $\mbox{Re}(w)$ denotes the real part, $c$,
of a complex number $w = c+di$ and $a$ and $\overline{a}$ are determined by the equations 
$a +  \overline{a} =  f_1$ and 
$a \lambda+ \overline{a}\overline{\lambda} = f_2$,
and in turn by the starting values $f_1$ and $f_2$. 
If we represent complex numbers by points in the plane, multiplication of
two numbers entails adding up the corresponding angles they form with
the positive $X$-axis. If the radiant $C^*_\delta\left(C^*_\delta(1-\delta)^2-4(1+\delta)\right)$
of $\lambda$  is negative, $\lambda$ is  not real and its angle is not equal to 0. 
Consequently, in this case there exists a smallest natural number $s$ such that 
$a\lambda^s$ lies in the left halfplane $\{X< 0\}$, so that $f_s$ becomes negative.
The roots of the radiant are $0$ and 
$4\frac{(1+\delta)}{(1-\delta)^2}$ which shows that 
$f^*_{n}$  gets negative if $C^*_\delta< 4\frac{(1+\delta)}{(1-\delta)^2}$ 
holds. The optimal strategy $F^*$ with $G_{(n,\delta)}(F^*)=C^*_\delta$  
has to be positive, therefore the overall competitive factor has to be at 
least $1+8\left(\frac{1+\delta}{1-\delta}\right)^2$,
which exactly matches the factor of the strategy described in \reftheo{first}.
\end{proof}

The proof also holds for $\delta=0$ which gives another proof of the 
factor $9$. Thus, line search with errors is generalized adequately.

\begin{prop}{second} 
In the presence of an error up to $\delta$ in
the standard multiplicative error model
$[\frac1{(1+\delta)}f, (1+\delta)f]$ for $\delta>0$ there is a competitive 
strategy that always meets the goal and achieves a factor of
$1+8(1+\delta)^4$. There is no strategy that achieves 
a better competitive factor.
\end{prop}
\begin{proof}
We follow the lines of \reftheo{first}, \reflem{exact} and \reftheo{opt}. 
In the standard multiplicative error model,
Equation \refequ{simplify} in the 
proof of \reftheo{first} reads 
$$\frac{|\pi_{\mathrm{onl}}|}{d} 
= 1 + 2(1+\delta)^2\frac{\sum_{i=1}^{2j+1}f_i}{f_{2j} - \delta(2+\delta)\sum_{i=1}^{2j-1}f_i + \varepsilon}$$
and it suffices to consider the functionals (compare to \refequ{func})
$$ G_{(n,\delta)}(F):=\frac{\sum_{i=1}^{n+1}f_i}{f_{n} - \delta(2+\delta)\sum_{i=1}^{n-1}f_i}\, .
$$
Analogously, the best doubling strategy $f_i=\alpha^i$ can be found by 
simple analysis which gives $\alpha=2(1+\delta)^2$. The corresponding factor 
 $G_{(n,\delta)}(F)< 4\,(1+\delta)^2$
can be computed as shown in (\ref{funcfac-equ}) 
which gives an overall factor of 
 $1+8(1+\delta)^4$ for the strategy $f_i=(2(1+\delta)^2)^i$. 

\medskip\noindent
The strategy proceeds in every iteration step at least by
\begin{eqnarray*}
{f_{n} \over 1+\delta } - \Delta_{n}
&=&
{\alpha^n \over 1+\delta} 
- {\delta\,(2+\delta) \over 1+\delta}\cdot
\sum_{i=1}^{n-1} \alpha^i\\
&=&
{\alpha^n \over 1+\delta}
- {\delta\,(2+\delta) \over 1+\delta}\cdot
{\alpha^n-\alpha
\over \alpha-1}\\
&=&
{(1+4\delta+2\delta^2)\,\alpha^n
-\delta\,(2+\delta)\,\alpha^n
+\delta\,(2+\delta)\,\alpha
\over (1+\delta)(1+4\delta+2\delta^2)}\\
&=& 
{(1+2\delta+\delta^2)\,\alpha^n
+\delta\,(2+\delta)\,\alpha
\over (1+\delta)(1+4\delta+2\delta^2)} > 0 \quad \mbox{for~} \delta > 0
\end{eqnarray*}
and the strategy will reach every goal.

It remains to show that the given strategy is optimal. \reflem{exact} also 
holds for the multiplicative error model. With the same techniques we 
adapt a given strategy such that the optimal factor holds in every step. 
the corresponding recurrence of an optimal strategy, see \refequ{rec}, 
is now given by 
$$
f^*_{n+1}=C^*_\delta f^*_{n-1}-C^*_\delta(1+\delta)^2f^*_{n-2}\, .   
$$
We consider the characteristic polynom which is
$X^2-C^*_\delta X+C^*_\delta(1+\delta)^2$, see (\ref{char-equ}). 
The polynom has the roots
\begin{eqnarray*}
\lambda, \overline{\lambda}  =  \frac12\left( C^*_\delta 
          \pm\sqrt{C^*_\delta(C^*_\delta-4(1+\delta)^2)}\right)
\end{eqnarray*}
Now with the same arguments as in the proof of \reftheo{opt} the radiant is 
non-negative for $C^*_\delta\geq 4(1+\delta)^2$ and the competitive 
factor has to be at least $1+8(1+\delta)^4$. \qed
\end{proof}

\sect{Error afflicted searching on \boldmath{$m$} rays}{mray}

The \agent\ is located at the common endpoint of $m$ infinite rays. The
target is located on one of the rays, but---as above---the \agent\ 
neither knows the ray containing the target nor the distance to the target.
It was shown by Gal~\cite{g-sg-80} that 
w.l.o.g.\  one can visit the rays in a cyclic order and with increasing depth. 
Strategies with this property are called {\em periodic} and {\em monotone}.   
More precisely, the values $f_i$ of a strategy $F$ denote the depth of a 
search in the $i$-th step. Further, $f_i$ and  $f_{i+m}$ visit the same 
ray, and $f_i<f_{i+m}$ holds. An optimal strategy is defined by 
$f_i=\left(\frac{m}{m-1}\right)^i$.

In the error afflicted setting, the start point of every iteration cannot drift away, since
the start point is the only point where all rays meet and the robot 
has to recognize this point. Otherwise we can not guarantee that all rays are 
visited.
Let us first assume that the error $\delta$ is known. 
Surprisingly, it will turn out that we do not have 
to distinguish whether $\delta$ is known or unknown to the
strategy.

\begin{theo}{mwaykomp}
Assume that an error afflicted robot with error range $\delta$ in 
the percentual error model is given.  
Searching for a target located on one of $m$ rays using 
a monotone and periodic strategy is competitive with an optimal factor of 
$$3+2\,\frac{1+\delta}{1-\delta}\left(\frac{m^m}{(m-1)^{m-1}}-1\right)$$
for $\delta<\frac{e-1}{e+1}$. 
\end{theo}

\begin{proof}
A periodic strategy $F$ with nominal values $f_1,f_2,f_3,\ldots$ is
monotone if $(1-\delta)f_k>(1+\delta)f_{k-m}$ holds. Now
let $\ell_i$ denote the distance covered by the error afflicted
agent in the step $i$. In analogy to the 
line case, we achieve the worst case, if the target is slightly missed in step $k$, but hit
in step $k+m$. This yields
$${|\pi_{\mathrm{onl}}| \over d}=
1+{2\sum_{i=1}^{k+m-1}\ell_i\over \ell_k+\varepsilon}\,.$$

This ratio achieves its maximum for $F$, if we  maximize 
every $\ell_i, i\neq k$ and take a worst case value for $\ell_k$. 
Therefore we set $\ell_k := (1-\beta)f_k$ and
$\ell_i := (1+\delta)f_i, i\neq k$ for $\beta\in[-\delta,\delta]$. 
For convenience we ignore $\varepsilon$ from now on.
We add $2(1+\delta-(1-\beta))f_k-2(1+\delta-(1-\beta))f_k)=0$ to the sum and obtain  
\begin{eqnarray}
{|\pi_{\mathrm{onl}}| \over d} & = & 
{1-2\,\frac{1+\delta-(1-\beta)}{1-\beta}+
2\,\frac{1+\delta}{1-\beta}\;\frac{\sum_{i=1}^{k+m-1}f_i}{f_k}}
\label{loes-equ}\\
& = & 3 + 2\,\frac{1+\delta}{1-\beta}\,\left(\frac{\sum_{i=1}^{k+m-1}f_i}{f_k}-1\right)\,.\nonumber
\end{eqnarray}
The functionals $G_k(F):=\frac{\sum_{i=1}^{k+m-1}f_i}{f_k}$ are identical 
to the functionals considered in the error-free $m$-ray search. 
From these results we know that the strategy  
$f_i=\left(\frac{m}{m-1}\right)^i$ 
gives the optimal upper bound $G_k(F)< \frac{m^m}{(m-1)^{m-1}}$, see 
\cite{bcr-sp-93,g-sg-80}. Now the adversary has the chance to maximize
$3+2\,\frac{1+\delta}{1-\beta}\left(\frac{m^m}{(m-1)^{m-1}}-1\right)$
over $\beta$ which obviously gives  $\beta=\delta$ 
Altogether, the factor and the optimality are proven. Note, 
that the optimal strategy is independent from $\delta$, thus there is no 
difference between known or unknown error range. 

We still have to ensure that $f_i=\left(\frac{m}{m-1}\right)^i$ is monotone
which means that the inequality
$(1-\delta)\left({m \over m-1}\right)^k > 
(1+\delta)\left({m \over m-1}\right)^{k-m}$
should be fulfilled, which in turn is equivalent to
$\delta < { \left({m \over m-1}\right)^m -1 \over 
\left({m \over m-1}\right)^m + 1} =: \delta_{\mathrm{max}}(m).$
Since $\delta_{\mathrm{max}}(m) 
\longrightarrow_{m \rightarrow \infty}
\frac {e - 1}{e + 1} \approx 0.4621$, we know that the best strategy
is given by $f_i=\left(\frac{m}{m-1}\right)^i$ if $\delta < 0.4621$ holds. 
\qed
\end{proof}

\begin{prop}{mwaykomp}
Assume that an error afflicted robot with error range $\delta$ in
the standard multiplicative error model is given.  
Searching for a target located on one of $m$ rays using 
a monotone and periodic strategy is competitive with an optimal factor of 
$3 + 2(1+\delta)^2\left(\frac{m^m}{(m-1)^{m-1}}-1\right)$
for $\delta<\sqrt{e}-1$. 
\end{prop} 
\begin{proof}
With the same arguments as in the proof of \reftheo{mwaykomp} the 
worst case ratio
$$
{|\pi_{\mathrm{onl}}| \over d}=
1+{2\sum_{i=1}^{k+m-1}\ell_i\over \ell_k+\varepsilon}\,,
$$
will be maximized if 
we set $\ell_k := \frac1{(1+\beta)}f_k$ and
$\ell_i := (1+\delta)f_i, i\neq k$ for $\beta\geq 0$.
We add $2\left((1+\delta)-\frac1{1+\beta}\right)f_k-2
          \left((1+\delta)-\frac1{1+\beta}\right)f_k$ 
to the sum  and achieve

\begin{eqnarray*}
{|\pi_{\mathrm{onl}}| \over d} & = & 
{1-2(1+\beta)\left((1+\delta)-\frac1{1+\beta}\right)+2(1+\beta)(1+\delta)\frac{\sum_{i=1}^{k+m-1}f_i}{f_k}}\\
& = & 3 + 2(1+\beta)(1+\delta)\left(\frac{\sum_{i=1}^{k+m-1}f_i}{f_k}-1\right)\,.\label{opt2-equ}
\end{eqnarray*}
Analogously, the best strategy for the functional  $G_k(F):=\frac{\sum_{i=1}^{k+m-1}f_i}{f_k}$ is given by $f_i=\left(\frac{m}{m-1}\right)^i$
and the adversary has the chance to maximize 
$$3 + 2(1+\beta)(1+\delta)\left(\frac{m^m}{(m-1)^{m-1}}-1\right)$$
which gives $\beta=\delta$. 

Preserving for monotonicity means $\frac1{(1+\delta)}\left({m \over m-1}\right)^k > (1+\delta)\left({m \over m-1}\right)^{k-m}$
should be fulfilled, which in turn is equivalent to
$\delta < \sqrt{\left({m \over m-1}\right)^m} -1\leq \sqrt{e}-1\,.$ \qed
\end{proof}

\sect{Summary}{concl}
We have analyzed the standard doubling strategy for reaching a door along 
a wall in the presence of errors in movements.  
We showed that the robot is still able to reach the door if the
error $\delta$ is not greater
than  $\frac13$ (33 per cent on a single step). 
The competitive ratio of the doubling strategy is given by 
$8 \frac{1+\delta}{1-3\delta}+1.$
The error bound is rather big, so it can 
be expected that real robots will meet this error. 
If the maximal error is known in advance the strategy 
$f_i=\left(2\,\frac{1+\delta}{1-\delta}\right)^i$ is the optimal 
competitive strategy with a competitive factor of 
$1+8\left(\frac{1+\delta}{1-\delta}\right)^2$. It was shown that 
the analysis technique can be applied to different error models.

In case of $m$ rays the problem is easier to solve since the \agent\
detects the start point after each return.  
If the error $\delta$ is not greater than 
$\delta_{\mathrm{max}}(m) = { \left({m \over m-1}\right)^m -1 \over 
\left({m \over m-1}\right)^m + 1}$, which is less than 
$\frac{e-1}{e+1}\approx 0.46212$ for all $m$, the
standard $m$-ray doubling strategy with $f_i=\left(\frac{m}{m-1}\right)^i$ is 
the optimal periodic and monotone strategy and 
yields a factor 
$3+2\,\frac{1+\delta}{1-\delta}\left(\frac{m^m}{(m-1)^{m-1}}-1\right)$. 


%% file: appendix.tex
\begin{appendix}
\noindent{\LARGE \bfseries{APPENDIX}}

\section{Finding the optimal strategy in \reftheo{first}}

For a fixed $\delta$, $1$ and $2(1+\delta)$ is constant, and 
it is sufficient to find a strategy that minimizes 
\begin{equ}{GndS}
G_{(n,\delta)}(S):=
{\sum_{i=1}^{n+1}f_i \over (1-\delta)\,f_n - 2\delta\sum_{i=1}^{n-1}f_i}
\quad \mbox{~for~}\,n\geq 1
\end{equ} 

\noindent
and $G_{(0,\delta)}(S):={f_1 \over 1}$
where $S$ denotes the strategy $f_1,f_2,f_3,\ldots$.
$G_{(0,\delta)}(S)$ is the worst case after the first iteration step.

\medskip
Now, we are searching for a strategy $S_\alpha$ 
in the form $f_i=\alpha^i$
with a fixed $\alpha$, possibly depending on $\delta$, that
asymptotically minimizes 
\begin{eqnarray*}
G_{n,\delta)}(S_\alpha)
&=&
{\sum_{i=1}^{n+1}\alpha^i \over (1-\delta)\,\alpha^n 
- 2\delta\sum_{i=1}^{n-1}\alpha^i}\\
&=&
{{\alpha^{n+2} - \alpha\over \alpha-1} 
\over
(1-\delta)\,\alpha^n -2\delta\,{\alpha^{n} - \alpha\over \alpha-1}}\\
&=&
{\alpha^2 - {1\over \alpha^{n-1}}
\over
(\alpha-1)(1-\delta) -2\delta + {2\delta \over \alpha^{n-1}}}\\
&<&
{\alpha^2
\over
(1-\delta)\,\alpha-\delta-1}
\;=:\;H_\delta(\alpha)
\end{eqnarray*}

\noindent
To find a minimum of $H_\delta(\alpha)$ we derivate and find the roots 
\begin{eqnarray*}
H'_\delta(\alpha) & = &
{2\alpha\,((1-\delta)\,\alpha-\delta-1) - (1-\delta)\,\alpha^2
\over
((1-\delta)\,\alpha-\delta-1)^2}\\
& = &
{(1-\delta)\,\alpha^2-2\,(1+\delta)\,\alpha
\over
(1-\delta)^2\,\alpha^2-2\,(1-\delta^2)\,\alpha+(1+\delta)^2} = 0\\
&\Leftrightarrow&
(1-\delta)\,\alpha^2-2\,(1+\delta)\,\alpha = 0\\
&\Leftrightarrow& 
\alpha =0 \quad \vee \quad \alpha={2\,(1+\delta)\over 1-\delta}
\end{eqnarray*}

A strategy with $\alpha=0$ will not move the robot at all, so 
$\alpha=2\,{1+\delta\over 1-\delta}$ is the only
reasonable root. Note, that the denominator of 
$H'_\delta(2\,{1+\delta\over 1-\delta})$
yields $(1+\delta)^2\neq 0$ for $\delta\geq 0$. To test whether 
this $\alpha$ is a maximum or minimum,
we use the second derivative.
Since we want to evaluate $H''_\delta(\alpha)$ only for the roots of
the numerator of $H'_\delta(\alpha)$, we can use a simplified
form%
\footnote{The derivative of a function of type $f(x)={N(x)\over D(x)}$
is $f'(x)={D(x)\cdot N'(x)-N(x)\cdot D'(x)\over (D(x))^2}$.
If we want to evaluate $f'(x)$ only for the roots of
$N(x)$, the derivative simplifies to
$f'\big|_{N(x)=0}(x)={N'(x) \over D(x)}$.}:

$$H''_\delta\Big|_{N'(x)=0}(\alpha)= 
{2\,(1-\delta)\,\alpha -2\,(1+\delta) \over (1+\delta)^2}\;.$$
This yields ${2\over 1+\delta} > 0$ for 
$\alpha=2\,{1+\delta\over 1-\delta}$, so we have found a minimum.
\end{appendix}